\documentclass[twoside]{article}

\usepackage[T2A]{fontenc}
\usepackage[tbtags]{amsmath}
\usepackage{amsthm}
\usepackage{amsfonts,amssymb,mathrsfs}
\usepackage{graphicx}

\theoremstyle{plain}
\newtheorem{theorem}{Theorem}

\theoremstyle{definition}

\begin{document}

\title{On the ultrametric generated by random distribution of points in Euclidean
spaces of large dimensions with correlated coordinates}

\author{A.\,P.~Zubarev \thanks{Physics Department, Samara State Aerospace University, Moskovskoe shosse 34, 443123, Samara, Russia. E-mail: apzubarev@mail.ru}
\thanks{Physics and Chemistry Department, Samara State University of Railway Transport, Perviy Bezimyaniy pereulok 18, 443066, Samara, Russia.}}
\maketitle
\begin{abstract}
In a recent paper the author proved a~theorem to the effect that the matrix of normalized Euclidean distances on the set of specially
distributed random points in the $n$-dimensional Euclidean space~$\mathbb R^{n}$ with independent coordinates converges in probability as
$n\rightarrow\infty$ to an ultrametric matrix, the latter being completely determined by the expectations of conditional variances of random coordinates of points.
The main theorem of the present paper extends this result to the case of weakly correlated coordinates of random points. Prior to formulating and stating this
result we give two illustrative examples describing particular algorithms of generation of such ultrametric spaces.
\end{abstract}

\section{Introduction}

One of the principal problems encountered in data processing is the problem of structuring of objects.
As a~rule, each object is represented by some feature vector. Besides, the data set is a~mapping from the set of objects into
the space of feature vectors for these objects.
Feature-based structuring of objects can be performed, for example, using agglomerative algorithms of hierarchic
clustering. As a~result of such clustering of objects one obtains a~quantitative
hierarchic classification thereof, which is equivalent to the definition of some indexed hierarchic structure (or an~ultrametric structure) on the set of objects~\cite{RTV}.
Such an approach is conventional for problems of classification of objects in terms of the closeness of their feature vectors.
It is of considerable interest to examine the data sets relating to a~set of objects which are naturally equipped with an explicit or implicit ultrametric structure.
If such an ultrametric structure is explicit, then this structure is directly felt in the data sets and
may be easily identified by analyzing the metric matrix defined on the space of  feature vectors of the objects. Nevertheless,  if the ultrametric structure
is implicit, then such an approach may fail to detect it. One possible approach to identify implicit ultrametric
structures is based on the transition from the initial data set to a~different data set (of possibly much smaller dimension), which corresponds to the choice of
new effective variables representing the space of features of the objects.

The principal problem in examining the systems with ultrametric structures is to identify the features of objects responsible for indexing of these
structures. When applied to specific systems  (data sets), such problems are generally nontrivial and have no universal solution algorithms.
Existence of ultrametric structures in data sets pertaining to systems of various natures was examined in a number of papers (see, for example, \cite{M1,M2,M3}).

We recall the definition of an ultrametric space. A set $M=\left\{ x\right\} $
is called a~metric space if, for any pair of elements,
$x^{a}, x^{b}\in M$, there is a~distance function between
$d_{ab}=d(x^{a},x^{b})$ (a~metric) satisfying the following conditions for any $x^{a},\, x^{b},\, x^{c}\in M$:
\begin{equation}
1)\, d_{ab}\ge0;\;2)\, d_{ab}=0\,\Leftrightarrow\, a=b;\;3)\, d_{ab}=d_{ba};\;4)\, d_{ab}\le d_{ac}+d_{dc}.\label{metric}
\end{equation}
A metric $d_{ab}$ satisfying the strong triangle inequality
\begin{equation}
d_{ab}\le\max\{d_{ac},\: d_{bc}\}\label{ultrametric}
\end{equation}
is called an ultrametric. A space with ultrametric is called a~space with ultrametric structure or an ultrametric
space. Endowing some set of objects with an~ultrametric structure is equivalent to specifying an indexed hierarchic structure~\cite{RTV}.
An arbitrary real matrix $d=\{d_{ab}\}$ is a~\textit{metric matrix} if its entries satisfy conditions~(\ref{metric}); it is called an \textit{ultrametric matrix} if
its entries satisfy conditions  (\ref{metric}) and~(\ref{ultrametric}).

One area of considerable promise is that of the search of possible mechanisms explaining the appearance of ultrametric structures in
systems of various natures. An efficient approach for attacking such problems is the use of the methods of analysis on ultrametric
spaces (or the ultrametric analysis). The ultrametric approach has been a~useful
tool for a~long time for solving various problems in the area of classification of objects and
information processing  of data sets~\cite{RTV}. The ultrametric analysis has been considerably advanced during the last 30 years---it received a~great impetus from the pioneering works of researches from the scientific school of Academician V.\,S.~Vladimirov, whose efforts were later taken up by
a~number of research works from different scientific schools (for an overview, consult, for example,~\cite{ALL}).
This relatively new research field is now known as the ``$p$-adic and ultrametric mathematical physics'' and is blessed by a~number of books
and an immense number of research articles in the area of $p$-adic analysis,
$p$-adic mathematical physics and their applications to modeling in various areas of physics, biology, computer science, sociology, physiology, etc.\ (see \cite{ALL,S,VVZ} and the references given therein).

In a number of studies dedicated to the application of methods of the ultrametric analysis
to real systems it was noted several times (see, for example, \cite{RTV,M3,R1,Hall,M4}) that the
correlation coefficients of sparse data sets of large
dimension have ultrametric properties. More exactly, it was shown that the matrices of normalized Euclidean distances between
randomly distributed points in a~multivariate space become close to ultrametric matrices as the dimension
of the space increases. The probabilistic justification of this fact based on the laws of large numbers was put forward in~\cite{Hall,M4}
for groups of random points with the same distribution in a~multivariate space.
In the recent paper~\cite{Z1} we formulated and proved a~general theorem stating that
the matrix of normalized Euclidean distances on the set of
specially distributed random points in the $n$-dimensional
Euclidean space $\mathbb R^n$ with independent  coordinates converges in probability
as $n\rightarrow\infty$ to the ultrametric matrix. The entries of this ultrametric matrix were given explicitly, and moreover,
were shown to be completely determined by the expectations of the conditional variances of the coordinates of random points.
In the present paper we extend the results of~\cite{Z1} and formulate and prove an analogous
theorem for the case of correlated coordinates of random points.

The paper is organized as follows. In the next section we, following~\cite{Z1}, give an illustrative example describing one particular
algorithm for constructing an ultrametric space by generating independent randomly distributed points in the $n$-dimensional  Euclidean
space with independent coordinates. In Section~3 we provide a~different illustrative example of the construction of an ultrametric space
by generating random points in the $n$-dimensional  Euclidean space,
in which the coordinates of random points are correlated in a~special way. In Section~4 we formulate and prove the main
theorem to the effect that the metric matrix of the set of independent
random points with correlated coordinates with special distribution
in the $n$-dimensional Euclidean space converges in probability, under certain conditions,
as $n\rightarrow\infty$ to the ultrametric matrix. Here, the ultrametric matrix is given in an explicit form, which is determined by the expectations of conditional variances
of random coordinates of points.

\section{Illustrative example~I. Independent coordinates}

We shall consider the following algorithm for construction of a~generation of independent
randomly distributed points in the $n$-dimensional Euclidean space with independent coordinates. This algorithm comprises the following  $N$-step procedure.

Let $p_{1}$, $p_{2}$, $\dots$, $p_{N}$ be a~sequence of natural numbers. At the first step we generate $p_{1}$ independent random
points $x^{(a_{1})}=\left(x_{1}^{(a_{1})},x_{2}^{(a_{1})},\ldots,x_{n}^{(a_{1})}\right)$
($a_{1}=1,2,\ldots,p_{1}$) in the $n$-dimensional space $\mathbb R^n$ with
normal distribution $\mathcal{N}\left(0,\sigma_{1}\right)$
for each independent coordinate. The coordinates of different points
are also assumed to be independent. At the second step we generate $p_{1}p_{2}$
independent random points $x^{(a_{1}a_{2})}=$$\left(x_{1}^{(a_{1}a_{2})},x_{2}^{(a_{1}a_{2})},\ldots,x_{n}^{(a_{1}a_{2})}\right)$
(here, $a_{1}=1,2,\ldots,p_{1}$, $a_{2}=1,2,\ldots,p_{2}$ and $a_{1}a_{2}$ is a~two-dimensional index) in~$\mathbb R^n$ with normal distribution $\mathcal{N}\left(x_{i}^{(a_{1})},\sigma_{2}\right)$
for the $i$th coordinate. Next, at the third step we generate $p_{1}p_{2}p_{3}$
independent random points $\left(x_{1}^{(a_{1}a_{2}a_{3})},x_{2}^{(a_{1}a_{2}a_{3})},\right.$$\left.\ldots,x_{n}^{(a_{1}a_{2}a_{3})}\right)$
($a_{1}=1,2,\ldots,p_{1}$, $a_{2}=1,2,\ldots,p_{2}$, $a_{3}=1,2,\ldots,p_{3}$)
in~$\mathbb R^n$ with normal distribution  $\mathcal{N}\left(x_{i}^{(a_{1}a_{2})},\sigma_{3}\right)$
for the $i$th coordinate, and so on. We shall repeat $N$~times this procedure of generation
of random  points. At the last $N$th step we generate
$p_{1}p_{2}\cdots p_{N}$ independent random  points $x^{(a_{1}a_{2}\cdots a_{N})}=$
$\left(x_{1}^{(a_{1}a_{2}\cdots a_{N})},x_{2}^{(a_{1}a_{2}\cdots a_{N})},\right.$
$\left.\ldots,x_{n}^{(a_{1}a_{2}\cdots a_{N})}\right)$ ($a_{1}=1,2,\ldots,p_{1}$,
$a_{2}=1,2,\ldots,p_{2}$, $\ldots,$ $a_{N}=1,2,\ldots,p_{N}$) in~$\mathbb R^n$ with normal distribution
$\mathcal{N}\left(x_{i}^{(a_{1}a_{2}\cdots a_{N-1})},\sigma_{N}\right)$ for the $i$th coordinate.

The set of points $M_{n}^{(N)}=\left\{ x^{(a_{1}a_{2}\cdots a_{N})}\right\} $ forms a~metric space with the metric
\begin{equation}
d_{n}\left(x^{(a_{1}a_{2}\cdots a_{N})},x^{(b_{1}b_{2}\cdots b_{N})}\right)=\dfrac{1}{\sqrt{n}}\sqrt{\sum_{i=1}^{n}\left(x_{i}^{(a_{1}a_{2}\cdots a_{N})}-x_{i}^{(b_{1}b_{2}\cdots b_{N})}\right)^{2}}.\label{d_n}
\end{equation}
The natural question here is how close is the metric matrix~(\ref{d_n}) to the ultrametric one for various~$n$? Below we shall give the results of
numerical simulations in which the number of steps of the procedure is $N=3$.

We let $d_{n}^{\left(p_{1},p_{2},\dots,p_{N}\right)}\left(N\right)$
denote the metric matrix that was numerically generated in accordance with the above procedure,
where $N$~is the number of steps in the procedure, $p_{i}$, $i=1,2,\dots,N$, is the number of points at the $i$th step of the procedure. Below we give the results of
calculation of an arbitrary random realization of the metric matrix $d_{n}^{\left(p_{1},p_{2},p_{3}\right)}\left(3\right)$
with fixed values of the variance $\sigma_{1}=\sigma_{2}=\sigma_{3}=10$
and when the dimensions~$n$ of the space~$\mathbb R^n$ are, respectively, $n=10$, $n=10^{2}$ and $n=10^{3}$.
\[
d_{10}^{\left(2,2,2\right)}\left(3\right)=\left(\begin{array}{cccccccc}
0 & 9.33 & 19.52 & 19.01 & 26.37 & 27.45 & 21.24 & 29.58\\
9.33 & 0 & 16.28 & 16.21 & 25.49 & 24.86 & 20.88 & 28.60\\
19.23 & 16.28 & 0 & 7.41 & 29.44 & 28.84 & 27.86 & 31.32\\
19.01 & 16.21 & 7.41 & 0 & 25.08 & 24.58 & 24.81 & 27.92\\
26.37 & 25.49 & 29.44 & 25.08 & 0 & 8.76 & 20.26 & 19.27\\
27.45 & 24.86 & 28.84 & 24.58 & 8.76 & 0 & 17.55 & 17.29\\
21.24 & 20.88 & 27.86 & 24.81 & 20.26 & 17.55 & 0 & 16.60\\
29.58 & 28.60 & 31.31 & 27.92 & 19.27 & 17.29 & 16.60 & 0
\end{array}\right),
\]

\[
d_{100}^{\left(2,2,2\right)}\left(3\right)=\left(\begin{array}{cccccccc}
0 & 13.78 & 19.41 & 19.93 & 27.80 & 27.79 & 24.85 & 25.07\\
13.78 & 0 & 19.04 & 20.75 & 28.83 & 28.40 & 26.13 & 25.53\\
19.41 & 19.04 & 0 & 14.78 & 29.68 & 28.17 & 26.57 & 26.33\\
19.93 & 20.75 & 14.78 & 0 & 30.91 & 30.00 & 27.22 & 26.51\\
27.80 & 28.83 & 29.68 & 30.91 & 0 & 12.93 & 19.01 & 20.07\\
27.79 & 28.40 & 28.17 & 30.00 & 12.93 & 0 & 18.23 & 19.45\\
24.86 & 26.13 & 26.58 & 27.22 & 19.01 & 18.23 & 0 & 14.20\\
25.07 & 25.53 & 26.33 & 26.51 & 20.07 & 19.45 & 14.20 & 0
\end{array}\right),
\]

\[
d_{1000}^{\left(2,2,2\right)}\left(3\right)=\left(\begin{array}{cccccccc}
0 & 14.13 & 19.59 & 19.86 & 24.60 & 23.93 & 25.14 & 24.90\\
14.13 & 0 & 19.93 & 19.93 & 24.56 & 24.31 & 25.41 & 24.89\\
19.59 & 19.93 & 0 & 14.22 & 24.69 & 24.74 & 24.92 & 24.80\\
19.86 & 19.93 & 14.22 & 0 & 23.75 & 24.04 & 24.28 & 24.19\\
24.60 & 24.56 & 24.69 & 23.75 & 0 & 14.43 & 19.52 & 19.88\\
23.93 & 24.31 & 24.74 & 24.04 & 14.43 & 0 & 19.37 & 19.46\\
24.14 & 25.41 & 24.92 & 24.28 & 19.52 & 19.37 & 0 & 14.01\\
24.90 & 24.89 & 24.80 & 24.19 & 19.88 & 19.46 & 14.01 & 0
\end{array}\right).
\]

It is seen that with increasing $n$ the random realization of the matrix $d_{n}^{\left(2,2,2\right)}\left(3\right)$
becomes more and more close to the ultrametric matrix.
Using the results of~\cite{Z1} one may show that as $n\rightarrow\infty$
the random realization of the matrix $d_{n}^{\left(2,2,2\right)}\left(3\right)$
will approach in probability the ultrametric matrix of the form
\[
d_{\infty}^{\left(2,2,2\right)}\left(3\right)=10\cdot\left(\begin{array}{cccccccc}
0 & \sqrt{2} & 2 & 2 & \sqrt{6} & \sqrt{6} & \sqrt{6} & \sqrt{6}\\
\sqrt{2} & 0 & 2 & 2 & \sqrt{6} & \sqrt{6} & \sqrt{6} & \sqrt{6}\\
2 & 2 & 0 & \sqrt{2} & \sqrt{6} & \sqrt{6} & \sqrt{6} & \sqrt{6}\\
20 & 20 & \sqrt{2}\cdot10 & 0 & \sqrt{6} & \sqrt{6} & \sqrt{6} & \sqrt{6}\\
\sqrt{6} & \sqrt{6} & \sqrt{6} & \sqrt{6} & 0 & \sqrt{2} & 2 & 2\\
\sqrt{6} & \sqrt{6} & \sqrt{6} & \sqrt{6} & \sqrt{2} & 0 & 2 & 2\\
\sqrt{6} & \sqrt{6} & \sqrt{6} & \sqrt{6} & 2 & 2 & 0 & \sqrt{2}\\
\sqrt{6} & \sqrt{6} & \sqrt{6} & \sqrt{6} & 2 & 2 & \sqrt{2} & 0
\end{array}\right)
\]

For any finite metric space~$\mathcal{U}$ one may define the so-called  ``degree of ultrametricity'' or the ``ultrametric measure'' of the space~$\mathcal{U}$ as some
value, which quantitatively determines the closeness of the metric matrix of the space~$\mathcal{U}$ to the ultrametric matrix. Various approaches to the definition of the
degree of ultrametricity were discussed in a~series of papers \cite{M1,R1,Z1,Lerman,M5,Missarov}.
In the present paper we shall use the definition proposed in~\cite{Missarov}. We recall this definition.

Let $\mathcal{U}$~be
a~finite metric space with elements $x^{a},\: a=1,2,\ldots,N$ and let  $d_{ab}\equiv d\left(x^{a},x^{b}\right)$
be a~metric on~$\mathcal{U}$. For any three points (triangle) $\left\{ x^{a},x^{b},x^{c}\right\} $ in~$\mathcal{U}$ we consider the function
\[
u(a,b,c)=2\dfrac{\mathrm{mid}\left\{ d_{ab},d_{bc},d_{ca}\right\} }{\max\left\{ d_{ab},d_{bc},d_{ca}\right\} }-1,
\]
where $\mathrm{mid}\left\{ d_{ab},d_{bc},d_{ca})\right\} $ is the length of the middle-length side of the triangle $\left\{ x^{a},x^{b},x^{c}\right\} $.
The quantity $u(a,b,c)$ will be called the \emph{degree of ultrametricity} of the triangle $\left\{ x^{a},x^{b},x^{c}\right\} $. The \emph{degree of
ultrametricity of a~metric space} $\mathcal{U}$ is the number~$U$, which is defined as the average of $u(a,b,c)$ over all possible triangles in~$\mathcal{U}$
that are distinct from $\left\{ x^{a},x^{b},x^{c}\right\} $:
\[
U=\frac{3!(N-3)!}{N!}\sum_{a=1}^{N}\sum_{b=a+1}^{N}\sum_{c=b+1}^{N}u(a,b,c).
\]
We note that $0\leq U\leq1$, where $U=1$ if $\mathcal{U}$ is an ultrametric space.

Figure~\ref{fig1} shows the pointwise dependence of the degree of ultrametricity~$U$ for the random realization of the  space of
points generated by the above procedure versus the dimension~$n$ of the Euclidean space~$\mathbb R^n$. The following parameters were
chosen: $N=3$, $p_{1},=p_{2}=p_{3}=2$, $\sigma_{1}=\sigma_{2}=\sigma_{3}=10$.
In this graph, to each value of~$n$ (in the range between $n=4$ and $n=10^{3}$) there corresponds one point associated with
one arbitrary realization of the ultrametric matrix $d_{n}^{\left(2,2,2\right)}\left(3\right)$.

\begin{center}
\begin{figure}[b]
\includegraphics[width=16cm]{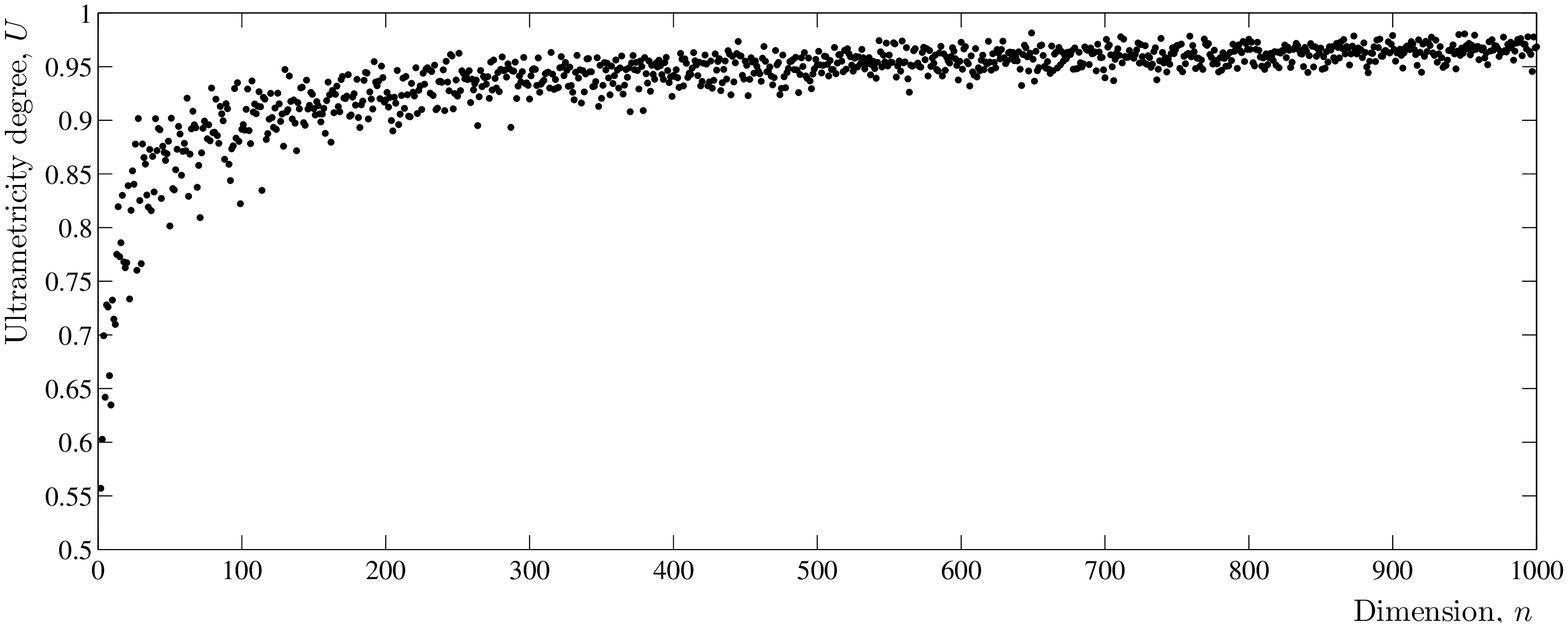}
\protect\caption{The degree of ultrametricity $U$ of the random realization of the metric space of points versus the dimension~$n$ of the Euclidean space~$\mathbb R^n$ with the following parameters of the generation
procedure: $N=3$, $p_{1},=p_{2}=p_{3}=2$, $\sigma_{1}=\sigma_{2}=\sigma_{3}=10$. }

\label{fig1}
\end{figure}

\par\end{center}

\section{Illustrative example II. Dependent coordinates}

In this section we shall consider the algorithm of generation of a~random points
with dependent coordinates in the $n$-dimensional Euclidean space; this algorithm is
a~modification of the algorithm considered in the previous section. In the framework of this modification
it is ensured that any two of $n$~coordinates of each random point
are nontrivially correlated. For our purposes it will be convenient to use the so-called hierarchic (cluster)
correlation of coordinates of each random point, the construction thereof will be described later.
In this construction the entire family of coordinates of each  random point is partitioned into hierarchically nested groups, which will be
called clusters. For such a~partition each coordinate can be associated with an element of some finite set~$M$ (a~finite ultrametric space).
Under this approach we shall assume that the covariance of two coordinates is defined by
the ultrametric distance between the elements of the space~$M$ that corresponds to these coordinates. It is worth pointing out that the ultrametric
on the space~$M$ associated with the set of coordinates of the  $n$-dimensional Euclidean space~$\mathbb R^n$ has no relation to the
ultrametric that arises on the set of random points in~$\mathbb R^n$.
The only reason for us to introduce an ultrametric on the space~$M$ is to provide for a~nontrivial
correlation  between the coordinates of each random point of~$\mathbb R^n$.

We set $n=m^{k}$, where $m$ and $k$ are positive integer numbers. The entire set of coordinates $\left\{ x_{1},x_{2},\ldots,x_{n}\right\} $
will be called the zero level cluster. We partition the zero level cluster $\left\{ x_{1},x_{2},\ldots,x_{n}\right\} $
into~$m$ groups
\[
\left\{ x_{1},x_{2},\ldots,x_{m^{k-1}}\right\} ,\,\,\left\{ x_{m^{k-1}+1},x_{m^{k-1}+2},\ldots,x_{2m^{k-1}}\right\} ,\ldots,
\]
\[
\left\{ x_{\left(m-1\right)m^{k-1}+1},x_{\left(m-1\right)m^{k-1}+2},\ldots,x_{m^{k}}\right\} ,
\]
which will be called the first level clusters and which will be indexed by an index $i_{1}=1,2,\ldots,m$. In turn, we partition each first level cluster
into $m$~subgroups, which will be called the second level clusters. For example, we partition the cluster $\left\{ x_{1},x_{2},\ldots,x_{m^{k-1}}\right\} $
into the subclusters
\[
\left\{ x_{1},x_{2},\ldots,x_{m^{k-2}}\right\} ,\,\,\left\{ x_{m^{k-2}+1},x_{m^{k-2}+2},\ldots,x_{2m^{k-2}}\right\} ,\ldots,
\]
\[
\left\{ x_{\left(m-1\right)m^{k-2}+1},x_{\left(m-1\right)m^{k-1}+2},\ldots,x_{m^{k-1}}\right\} .
\]
Second level clusters will be indexed by a~two-dimensional index $\left(i_{1}i_{2}\right)$,
$i_{1},i_{2}=1,2,\ldots m$, where $i_{1}$ is the number of a~first level cluster which contains the second level cluster, $i_{2}$ is the number
 of a~second level cluster inside the $i_{1}$th first level cluster.
Next, we partition each second level cluster into $m$~subgroups, which will be called third level clusters. Third level clusters
will be indexed by a~three-dimensional index  $\left(i_{1}i_{2}i_{3}\right)$, $i_{1},i_{2},i_{3}=1,2,\ldots m$. We continue this process of partition
until we get the $(k-1)$th level clusters of which each contains  $m$~coordinates. Thus, under such a~partition all the coordinates
are united into hierarchically nested clusters. Moreover, each coordinate which lies in an $i_{1}$th first level cluster,
in an $i_{2}$th  second level cluster, $\ldots,$ and in an $i_{k-1}$th  $(k-1)$th level cluster can be indexed by a~multiindex  $\alpha=\left(i_{1},i_{2},\ldots,i_{k}\right)$,
$i_{1},i_{2},\ldots,i_{k}=1,2,\ldots m$.

Let $\xi_{i_{1}i_{2}\ldots i_{k}}$, $\xi_{i_{1}i_{2}\ldots i_{k-1}}$, $\ldots,$$\xi_{i_{1}i_{2}}$, $\xi_{i_{1}}$, be families of independent random
variables distributed according to the normal distribution law $\mathcal{N}\left(0,\sigma\right)$.
We define the random coordinates of a~point~$x$ in~$\mathbb R^n$, $n=m^{k}$, as
\begin{equation}
x_{i_{1}i_{2}\ldots i_{k}}=\xi_{i_{1}i_{2}\ldots i_{k}}+\lambda^{-1}\xi_{i_{1}i_{2}\ldots i_{k-1}}+\lambda^{-2}\xi_{i_{1}i_{2}\ldots i_{k-2}}+\cdots+\lambda^{-k+2}\xi_{i_{1}i_{2}}+\lambda^{-k+1}\xi_{i_{1}},\label{x}
\end{equation}
where $\lambda$ is the control correlation parameter of  coordinates of random points. It is easily seen that the expectation and variance
of~(\ref{x}) are as follows
\[
\mathsf{E}\left[x_{i_{1}i_{2}\ldots i_{k}}\right]=0,\qquad \mathsf{Var}\left[x_{i_{1}i_{2}\ldots i_{k}}\right]=\sigma^{2}\sum_{j=0}^{k-1}\lambda^{-2j}=\sigma^{2}\dfrac{1-\lambda^{-2k}}{1-\lambda^{-2}}.
\]
One can also easily calculate the covariance of two random coordinates of~(\ref{x})
of the form $x_{i_{1}i_{2}\ldots i_{r-1}i_{r}i_{r+1}\ldots i_{k}}$ and $x_{i_{1}i_{2}\ldots i_{r-1}j_{r}j_{r+1}\ldots j_{k}}$
for various indexes $i_{r}i_{r+1}\ldots i_{k}$ and $j_{r}j_{r+1}\ldots j_{k}$:
\[
\mathsf{cov}\left[x_{i_{1}i_{2}\ldots i_{r-1}i_{r}i_{r+1}\ldots i_{k}},x_{i_{1}i_{2}\ldots i_{r-1}j_{r}j_{r+1}\ldots j_{k}}\right]=\sigma^{2}\sum_{j=r}^{k-1}\lambda^{-2j}=\sigma^{2}\lambda^{-2r}\dfrac{1-\lambda^{-2\left(k-r\right)}}{1-\lambda^{-2}}.
\]
We note that, for $\lambda>1$,
\[
\mathsf{Var}\left[x_{i_{1}i_{2}\ldots i_{k}}\right]\underset{{\scriptstyle n\rightarrow\infty}}{\longrightarrow}\dfrac{\sigma^{2}}{1-\lambda^{-2}},
\]
\[
\mathsf{cov}\left[x_{i_{1}i_{2}\ldots i_{r-1}i_{r}i_{r+1}\ldots i_{k}},x_{i_{1}i_{2}\ldots i_{r-1}j_{r}j_{r+1}\ldots j_{k}}\right]\underset{{\scriptstyle n\rightarrow\infty}}{\longrightarrow}\dfrac{\sigma^{2}\lambda^{-2r}}{1-\lambda^{-2}}.
\]

Let $p_{1}$, $p_{2}$, $\dots$, $p_{N}$ be a set of natural numbers. We generate $p_{1}$ independent random points $y^{(a_{1})}$
($a_{1}=1,2,\ldots,p_{1}$ ) in the $n$-dimensional space~$\mathbb R^n$ so that the coordinates $y_{i_{1}i_{2}\ldots i_{k}}^{(a_{1})}$ of each point
$x^{(a_{1})}$ are defined as $y_{i_{1}i_{2}\ldots i_{k}}^{(a_{1})}=x_{i_{1}i_{2}\ldots i_{k}}^{(a_{1})}$,
where $x_{i_{1}i_{2}\ldots i_{k}}^{(a_{1})}$ are defined by~(\ref{x}).
Next, we generate $p_{1}p_{2}$ independent random points $y^{(a_{1}a_{2})}$ (here, $a_{1}=1,2,\ldots,p_{1}$, $a_{2}=1,2,\ldots,p_{2}$, and $a_{1}a_{2}$
is a~two-dimensional index) in~$\mathbb R^n$. Besides each random coordinate
$y_{i_{1}i_{2}\ldots i_{k}}^{(a_{1}a_{2})}$ is defined as $y_{i_{1}i_{2}\ldots i_{k}}^{(a_{1}a_{2})}=x_{i_{1}i_{2}\ldots i_{k}}^{(a_{1}a_{2})}+y_{i_{1}i_{2}\ldots i_{k}}^{(a_{1})}$,
where $x_{i_{1}i_{2}\ldots i_{k}}^{(a_{1}a_{2})}$ are defined by~(\ref{x}).
Next, we generate $p_{1}p_{2}p_{3}$ independent random points
$y^{(a_{1}a_{2}a_{3})}$ ($a_{1}=1,2,\ldots,p_{1}$, $a_{2}=1,2,\ldots,p_{2}$,
$a_{3}=1,2,\ldots,p_{3}$) in~$\mathbb R^n$, for which
$y_{i_{1}i_{2}\ldots i_{k}}^{(a_{1}a_{2}a_{3})}=x_{i_{1}i_{2}\ldots i_{k}}^{(a_{1}a_{2}a_{3})}+y_{i_{1}i_{2}\ldots i_{k}}^{(a_{1}a_{2})}$,
where $x_{i_{1}i_{2}\ldots i_{k}}^{(a_{1}a_{2}a_{3})}$ are defined as in~(\ref{x}), and so~on. We repeat $N$~times this procedure of generation of random
points. At the next step, we generate $p_{1}p_{2}\cdots p_{N}$
independent random points ($a_{1}=1,2,\ldots,p_{1}$, $a_{2}=1,2,\ldots,p_{2}$,
$\ldots,$ $a_{N}=1,2,\ldots,p_{N}$) in~$\mathbb R^n$ for which $y_{i_{1}i_{2}\ldots i_{k}}^{(a_{1}a_{2}\cdots a_{N})}=x_{i_{1}i_{2}\ldots i_{k}}^{(a_{1}a_{2}\cdots a_{N})}+y_{i_{1}i_{2}\ldots i_{k}}^{(a_{1}a_{2}\cdots a_{N-1})}$,
where $x_{i_{1}i_{2}\ldots i_{k}}^{(a_{1}a_{2}\cdots a_{N})}$ are defined as in~(\ref{x}). The set of points $\mathcal{U}_{n}^{(N)}=\left\{ x^{(a_{1}a_{2}\cdots a_{N})}\right\} $
forms a~metric space with the metric
\begin{equation}
d_{n}\left(x^{(a_{1}a_{2}\cdots a_{N})},x^{(b_{1}b_{2}\cdots b_{N})}\right)\equiv\dfrac{1}{\sqrt{n}}\sqrt{\sum_{i_{1},i_{2},\ldots,i_{k}=1}^{m}\left(x_{i_{1}i_{2}\ldots i_{k}}^{(a_{1}a_{2}\cdots a_{N})}-x_{i_{1}i_{2}\ldots i_{k}}^{(b_{1}b_{2}\cdots b_{N})}\right)^{2}}.\label{d_n_corr}
\end{equation}

We next give the result of numerical simulations for the study of the dependence of the ultrametricity degree~$U$ of the metric  space $\mathcal{U}_{n}^{(N)}$
on the dimension~ $n$ of the Euclidean space $\mathbb R^n$ for various values of the parameter of correlation of coordinates~$\lambda$. Clearly,
in the limit $\lambda\rightarrow\infty$ we have the case of independent  coordinates examined in the previous section, and hence we shall be concerned with the case when
$\lambda$~is close to~1. We shall choose the following values~$n$ of the dimension: $n=2^{i}$, $i=2,\ldots,11$. Next, for each
such~$n$ we generate $10$ realizations of the $8\times8$ metric matrix
$d_{n}\left(x^{(a_{1}a_{2}a_{3})},x^{(b_{1}b_{2}b_{3})}\right)$,
$a_{i},b_{i}=1,2$ with the following fixed values of the parameters:
$m=2$, $N=3$, $p_{1}=p_{2}=p_{3}=2$, $\sigma=10$ and with  $\lambda=0.8$, $\lambda=1.2$,
$\lambda=2$ and $\lambda=10$, respectively. Next, for each realization we calculate the degree of ultrametricity~$U$, and then calculate the average~$\overline{U}$
of the degree of ultrametricity~$U$ over all $10$ realizations. Figure~\ref{fig2} depicts the pointwise graphs of~$\overline{U}\left(n\right)$
for various~$\lambda$. These calculations show that, for weakly correlated coordinates of random points ($\lambda>1$), the degree of ultrametricity~$U$ of random realizations of
the metric matrix $d_{n}\left(x^{(a_{1}a_{2}a_{3})},x^{(b_{1}b_{2}b_{3})}\right)$
tends in probability to~$1$ with increasing~$n$ (see the points corresponding to $\lambda=1.2$, $\lambda=2$, and $\lambda=10$ in  Fig.~\ref{fig2}).
Nevertheless, for strong correlations of coordinates of random points ($\lambda\leq1$),
the degree of ultrametricity in the random realizations of the metric matrix is not increasing with is increasing~ $n$  (see the points corresponding to $\lambda=0.8$ in Fig.~\ref{fig2}).
In the next section we shall formulate and prove the main  theorem, which gives
conditions on the correlation functions of coordinates of random points ensuring that the metric matrix of these points converges in probability
to the ultrametric matrix in the limit $n\rightarrow\infty$.

\begin{center}
\begin{figure}[b]
\includegraphics[width=16cm]{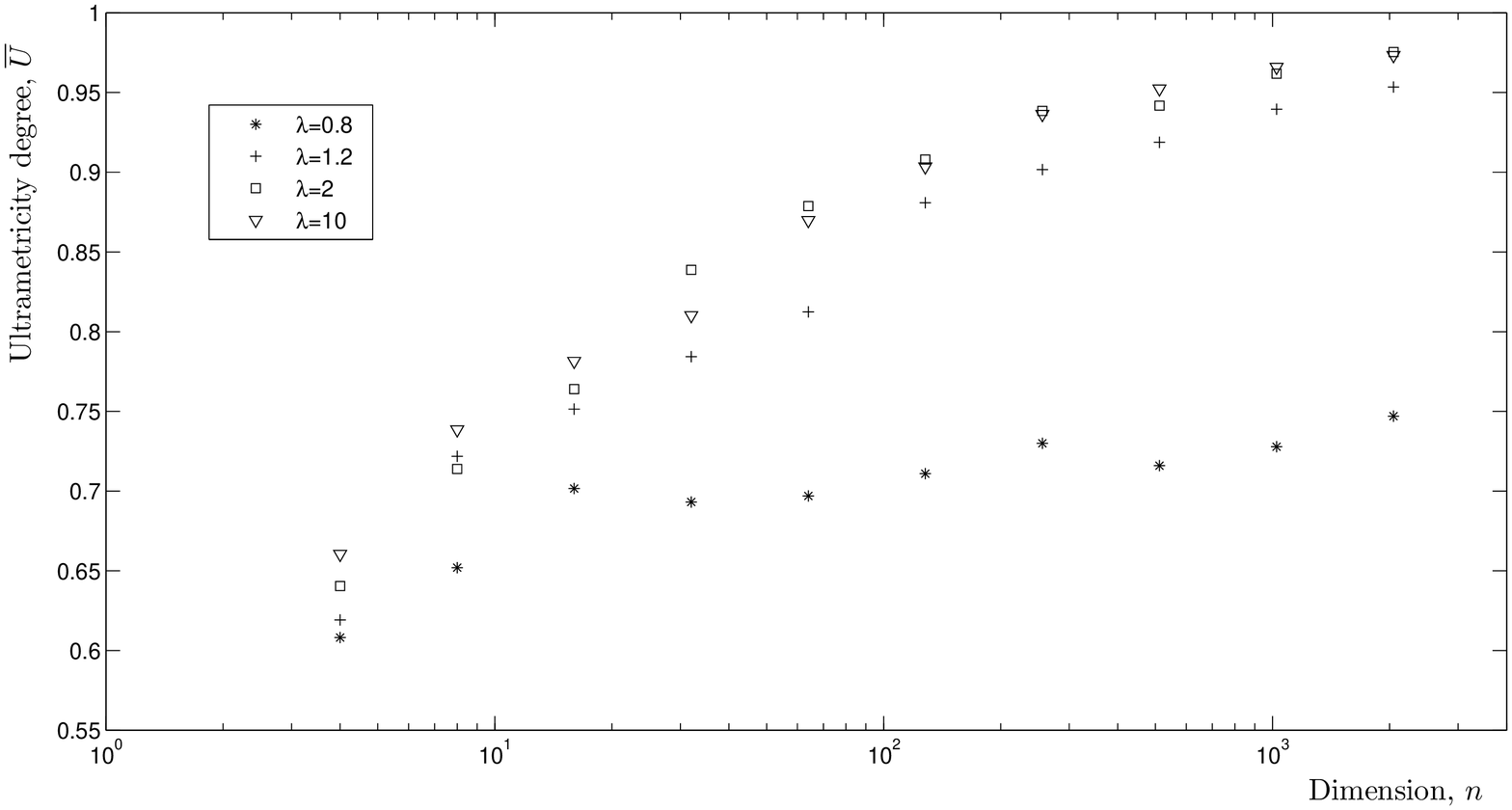}
\protect\caption{The degree of ultrametricity
$\overline{U}$ of the metric space $\mathcal{U}_{n}^{(N)}$, as averaged over 10 realizations, versus
the dimension~$n$ of the Euclidean space $\mathbb R^n$ with the parameters
of generation of random points $N=3$, $p_{1},=p_{2}=p_{3}=2$,
$\sigma_{1}=\sigma_{2}=\sigma_{3}=10$ and when the correlation parameter
of~coordinates $\lambda$ equals to $0.8$, $1.2$, $2$ and $10$. }

\label{fig2}
\end{figure}

\par\end{center}

\section{The main theorem}

First we need some results from probability that will be used to prove the main theorem.

Let $\left\{ \Omega,\Sigma,\mathrm{P}\right\} $ be a probability
space, where $\left\{ \Omega,\Sigma\right\}$ is measurable
space, $\mathrm{P}$ is probability measure. A~real random variable
$X$ is a~measurable mapping $X:\:\Omega\to R$. For any real random
variable $X=X(\omega)$ one may define the integral $\intop_{A}X\left(\omega\right)d\mathrm{P}\left(\omega\right)$,
$A\in\Sigma$. The expectation and the variance of~$X$
are defined, respectively, as $\mathsf{E}\left[X\right]=\intop_{\Omega}X\left(\omega\right)d\mathrm{P}\left(\omega\right)$
and $\mathsf{V}\left[X\right]=\mathsf{E}\left[X^{2}\right]-\left(\mathsf{E}\left[X\right]\right)^{2}$. Let $\Sigma^{(1)}\subset\Sigma$ be a $\sigma$-subalgebra
of $\Sigma$, then the conditional expectation $\mathsf{E}\left[X\left|\Sigma^{(1)}\right.\right]$
of a~real random variable~$X$ is a~random variable $Y$ such that $Y$~is $\Sigma^{(1)}$ measurable and, for all $A\in\Sigma^{(1)}$, $\intop_{A}X\left(\omega\right)d\mathrm{P}\left(\omega\right)=\intop_{A}Y\left(\omega\right)d\mathrm{P}\left(\omega\right)$,
and the conditional variance $\mathsf{Var}\left[X\left|\Sigma^{(1)}\right.\right]$
is defined as $\mathsf{Var}\left[X\left|\Sigma^{(1)}\right.\right]=\mathsf{E}\left[X^{2}\left|\Sigma^{(1)}\right.\right]-\left(\mathsf{E}\left[X\left|\Sigma^{(1)}\right.\right]\right)^{2}$.

\begin{theorem}
{\rm (Markov's theorem, \cite{Shiryaev})} Let $X_{1},X_{2},\ldots$
be a sequence of dependent random variables with finite expectations
$\mathsf{E}\left[X_{i}\right]\equiv m_{i}$, let  $\dfrac{\mathsf{Var}\left[\sum_{i=1}^{n}X_{i}\right]}{n^{2}}\rightarrow0$
as $n\rightarrow\infty$, and let $S_{n}=X_{1}+X_{2}+\ldots+X_{n}$.
Then $\dfrac{S_{n}}{n}\overset{\mathrm{P}}{\rightarrow}\dfrac{\left\langle S_{n}\right\rangle }{n}$;
i.e., for any $\varepsilon>0$ one has $\mathrm{P}\left\{ \left|\dfrac{S_{n}}{n}-\dfrac{\left\langle S_{n}\right\rangle }{n}\right|\geq\varepsilon\right\} \rightarrow0$
as $n\rightarrow\infty$ (the convergence in probability).
 \label{th1}
\end{theorem}

\begin{theorem}
{\rm (Slutsky's theorem, \cite{Sl,Bill})} If $X_{n}^{(1)}\overset{\mathrm{P}}{\rightarrow}X^{(1)}$,
$X_{n}^{(2)}\overset{\mathrm{P}}{\rightarrow}X^{(2)}$, $\ldots,$
$X_{n}^{(N)}\overset{\mathrm{P}}{\rightarrow}X^{(N)}$, and $h\left(x^{(1)},x^{(2)},\ldots,x^{(N)}\right)$
is a continuous function of $N$ variables, then 
$$h\left(X_{n}^{(1)},X_{n}^{(2)},\ldots,X_{n}^{(N)}\right)
\overset{\mathrm{P}}{\rightarrow}
h\left(X^{(1)},X^{(2)},\ldots,X^{(N)}\right).$$
 \label{th2}
\end{theorem}

The main theorem is as follows.
\begin{theorem}
 Let $\left\{ \Omega,\Sigma,\mathrm{P}\right\} $
be a~probability space, $\Sigma^{(n)}$ be an increasing sequence of
$\sigma$-subalgebras $\Sigma^{(1)}\subset\Sigma^{(2)}\subset\cdots\subset\Sigma^{(N)}\equiv\Sigma$.
Let $M_{n}\equiv\left\{ x^{(a_{1}a_{2}\cdots a_{N})}\right\} $ ($a_{1}=1,2,\ldots,p_{1}$,
$a_{2}=1,2,\ldots,p_{2}$, $\ldots,$ $a_{N}=1,2,\ldots,p_{N}$, and
$a_{1}a_{2}\cdots a_{N}$ is $N$-dimensional index) be sets of
$p_{1}p_{2}\cdots p_{N}$ independent random points $x^{(a_{1}a_{2}\cdots a_{N})}=\left(x_{1}^{(a_{1}a_{2}\cdots a_{N})},x_{2}^{(a_{1}a_{2}\cdots a_{N})},\ldots,x_{n}^{(a_{1}a_{2}\cdots a_{N})}\right)$
in $\mathbb R^n$ with generally dependent random coordinates. Assume that the conditional
expectations $\mathsf{E}\left[x_{i}^{(a_{1}a_{2}\cdots a_{N})}\left|\Sigma^{(k)}\right.\right]\equiv x_{i}^{(a_{1}a_{2}\cdots a_{k})}$
($k=1,2,\ldots,N-1$) are identical for all $a_{k+1},\: a_{k+2},\:\ldots,\: a_{N}$.
Next, assume that  $\mathsf{E}\left[\left(x_{i}^{(a_{1}a_{2}\cdots a_{N})}\right)^{m}\right]$,
$m=1,2,3,4$ and $\mathsf{cov}\left[x_{i}^{(a_{1}a_{2}\cdots a_{N})},x_{j}^{(a_{1}a_{2}\cdots a_{N})}\right]$
are uniformly bounded, $\mathsf{E}\left[x_{i}^{(a_{1}a_{2}\cdots a_{N})}\right]=m_{i}$,
$\mathsf{Var}\left[x_{i}^{(a_{1})}\right]=\sigma_{1}^{2}$,
$\mathsf{E}\left[\mathsf{Var}\left[x_{i}^{(a_{1}a_{2}\cdots a_{k})}\left|\Sigma^{(k)}\right.\right]\right]=\sigma_{k+1}^{2}$
($k=1,\ldots,N-1$) and the conditions
\[
\lim_{n\rightarrow0}\dfrac{1}{n^{2}}\sum_{\substack{i,j=1\\i<j}}^{n}\mathsf{cov}\left[\left(x_{i}^{(a_{1}a_{2}\cdots a_{N})}\right)^{l_{1}},\left(x_{j}^{(a_{1}a_{2}\cdots a_{N})}\right)^{l_{2}}\right]=0,\,\, l_{1},l_{2}=1,2
\]
 are hold. Then the metric on $M_{n}$
\[
d_{n}\left(x^{(a_{1}a_{2}\cdots a_{N})},x^{(b_{1}b_{2}\cdots b_{N})}\right)\equiv\dfrac{1}{\sqrt{n}}\sqrt{\sum_{i=1}^{n}\left(x_{i}^{(a_{1}a_{2}\cdots a_{N})}-x_{i}^{(b_{1}b_{2}\cdots b_{N})}\right)^{2}}
\]
 has the property
\[
d_{n}\left(x^{(a_{1}a_{2}\cdots a_{N})},x^{(b_{1}b_{2}\cdots b_{N})}\right)\overset{\mathrm{P}}{\longrightarrow}u_{a_{1}a_{2}\cdots a_{N},b_{1}b_{2}\cdots b_{N}}
\]
as $n\to\infty$, where
\[
u_{a_{1}a_{2}\cdots a_{N},b_{1}b_{2}\cdots b_{N}}=\sqrt{2}\left(\left(1-\delta_{a_{1}b_{1}}\delta_{a_{2}b_{2}}\cdots\delta_{a_{N}b_{N}}\right)\sigma_{N}^{2}+\right.
\]
\[
\left.+\left(1-\delta_{a_{1}b_{1}}\delta_{a_{2}b_{2}}\cdots\delta_{a_{N-1}b_{N-1}}\right)\sigma_{N-1}^{2}+\cdots+\left(1-\delta_{a_{1}b_{1}}\right)\sigma_{1}^{2}\right)^{\tfrac{1}{2}}
\]
 is an ultrametric $p_{1}p_{2}\cdots p_{N}\times p_{1}p_{2}\cdots p_{N}$--matrix.
 \label{th3}
\end{theorem}
\begin{proof}
First we prove by induction that
\[
\mathsf{E}\left[\left(x_{i}^{(a_{1}a_{2}\cdots a_{N})}-x_{i}^{(b_{1}b_{2}\cdots b_{N})}\right)^{2}\right]=
\]
\[
=2\left(\left(1-\delta_{a_{1}b_{1}}\delta_{a_{2}b_{2}}\cdots\delta_{a_{N}b_{N}}\right)\sigma_{N}^{2}+\right.
\]
\begin{equation}
\left.+\left(1-\delta_{a_{1}b_{1}}\delta_{a_{2}b_{2}}\cdots\delta_{a_{N-1}b_{N-1}}\right)\sigma_{N-1}^{2}+\cdots+\left(1-\delta_{a_{1}b_{1}}\right)\sigma_{1}^{2}\right).\label{ind}
\end{equation}
For the case $N=2$ we have
\[
\mathsf{E}\left[\left(x_{i}^{(a_{1}a_{2})}-x_{i}^{(b_{1}b_{2})}\right)^{2}\right]=
\]
\[
=\left(1-\delta_{a_{1}b_{1}}\delta_{a_{2}b_{2}}\right)
\mathsf{E}\left[\mathsf{E}\left[\left(x_{i}^{(a_{1}a_{2})}
-x_{i}^{(b_{1}b_{2})}\right)^{2}\left|\Sigma^{(1)}\right.\right]\right]=
\]
\[
=\left(1-\delta_{a_{1}b_{1}}\delta_{a_{2}b_{2}}\right)
\mathsf{E}\left[\mathsf{E}\left[\left(x_{i}^{(a_{1}a_{2})}\right)^{2}\left|\Sigma^{(1)}\right.\right]\right.+
\]
\[
+\left.\mathsf{E}\left[\left(x_{i}^{(b_{1}b_{2})}\right)^{2}\left|\Sigma^{(1)}\right.\right]
-2\mathsf{E}\left[x_{i}^{(a_{1}a_{2})}\left|\Sigma^{(1)}\right.\right]\mathsf{E}\left[x_{i}^{(b_{1}b_{2})}\left|\Sigma^{(1)}\right.\right]\right]=
\]
\[
=\left(1-\delta_{a_{1}b_{1}}\delta_{a_{2}b_{2}}\right)\mathsf{E}\left[\mathsf{Var}\left[\left(x_{i}^{(a_{1}a_{2})}\right)\left|\Sigma^{(1)}\right.\right]+\mathsf{Var}\left[\left(x_{i}^{(b_{1}b_{2})}\right)\left|\Sigma^{(1)}\right.\right]+\right.
\]
\[
\left.+\left(1-\delta_{a_{1}b_{1}}\delta_{a_{2}b_{2}}\right)\left(1-\delta_{a_{1}b_{1}}\right)\left(\mathsf{E}\left[\left(x_{i}^{(a_{1}a_{2})}\right)\left|\Sigma^{(1)}\right.\right]-\mathsf{E}\left[\left(x_{i}^{b{}_{1}b_{2})}\right)\left|\Sigma^{(1)}\right.\right]\right)^{2}\right]=
\]
\[
=\left(1-\delta_{a_{1}b_{1}}\delta_{a_{2}b_{2}}\right)\left(\mathsf{E}\left[\mathsf{Var}\left[\left(x_{i}^{(a_{1}a_{2})}\right)\left|\Sigma^{(1)}\right.\right]\right]+\mathsf{E}\left[\mathsf{Var}\left[\left(x_{i}^{(b_{1}b_{2})}\right)\left|\Sigma^{(1)}\right.\right]\right]\right)+
\]
\[
+\left(1-\delta_{a_{1}b_{1}}\right)\left(\mathsf{Var}\left[\mathsf{E}\left[x_{i}^{\left(a_{1}a_{2}\right)}\right]\left|\Sigma^{(1)}\right.\right]+\mathsf{Var}\left[\mathsf{E}\left[x_{i}^{\left(b_{1}b_{2}\right)}\right]\left|\Sigma^{(1)}\right.\right]\right)+
\]
\[
+\left(1-\delta_{a_{1}b_{1}}\right)\left(\mathsf{E}\left[\left(x_{i}^{(a_{1}a_{2})}\right)\right]
-\mathsf{E}\left[\left(x_{i}^{b{}_{1}b_{2})}\right)\right]\right)^{2}=
\]
\[
=2\left(\left(1-\delta_{a_{1}b_{1}}\delta_{a_{2}b_{2}}\right)\sigma_{2}^{2}+\left(1-\delta_{a_{1}b_{1}}\right)\sigma_{1}^{2}\right).
\]
We suppose that (\ref{ind}) is true for $N=k$ and then will show that (\ref{ind}) is true for $N=k+1$

\[
\mathsf{E}\left[\left(x_{i}^{(a_{1}a_{2}\cdots a_{k+1})}-x_{i}^{(b_{1}b_{2}\cdots b_{k+1})}\right)^{2}\right]=
\]
\[
=\left(1-\delta_{a_{1}b_{1}}\delta_{a_{2}b_{2}}\cdots\delta_{a_{k+1}b_{k+1}}\right)\mathsf{E}\left[\mathsf{E}\left[x_{i}^{(a_{1}a_{2}\cdots a_{k+1})}\left|\Sigma^{(k)}\right.\right]\right]=
\]
\[
=\left(1-\delta_{a_{1}b_{1}}\delta_{a_{2}b_{2}}\cdots\delta_{a_{k+1}b_{k+1}}\right)
\mathsf{E}\left[\mathsf{Var}\left[x_{i}^{(a_{1}a_{2}\cdots a_{k+1})}\left|\Sigma^{(k)}\right.\right]\right.+
\]
\[
+\left.\mathsf{Var}\left[x_{i}^{(b_{1}b_{2}\cdots b_{k+1})}\left|\Sigma^{(k)}\right.\right]\right]+
\]
\[
+\left(1-\delta_{a_{1}b_{1}}\delta_{a_{2}b_{2}}\cdots\delta_{a_{k}b_{k}}\right)\mathsf{E}\left[\left(x_{i}^{(a_{1}a_{2}\cdots a_{k})}-x_{i}^{(b_{1}b_{2}\cdots b_{k})}\right)^{2}\right]=
\]
\[
=2\left(1-\delta_{a_{1}b_{1}}\delta_{a_{2}b_{2}}\cdots\delta_{a_{k+1}b_{k+1}}\right)\sigma_{k+1}^{2}+
\]
\[
+\left(1-\delta_{a_{1}b_{1}}\delta_{a_{2}b_{2}}\cdots\delta_{a_{k}b_{k}}\right)\mathsf{E}\left[\left(x_{i}^{(a_{1}a_{2}\cdots a_{k})}-x_{i}^{(b_{1}b_{2}\cdots b_{k})}\right)^{2}\right]=
\]
\[
=2\left(\left(1-\delta_{a_{1}b_{1}}\delta_{a_{2}b_{2}}\cdots\delta_{a_{k+1}b_{k+1}}\right)\sigma_{k+1}^{2}+\right.
\]
\[
\left.+\left(1-\delta_{a_{1}b_{1}}\delta_{a_{2}b_{2}}\cdots\delta_{a_{k}b_{k}}\right)\sigma_{k}^{2}+\cdots+\left(1-\delta_{a_{1}b_{1}}\right)\sigma_{1}^{2}\right).
\]
This proves equation (\ref{ind}).

Next let us consider
\[
\mathsf{Var}\left[\sum_{i=1}^{n}\left(x_{i}^{(a_{1}a_{2}\cdots a_{N})}-x_{i}^{(b_{1}b_{2}\cdots b_{N})}\right)^{2}\right]=
\]
\[
=\sum_{i=1}^{n}\mathsf{Var}\left[\left(x_{i}^{(a_{1}a_{2}\cdots a_{N})}-x_{i}^{(b_{1}b_{2}\cdots b_{N})}\right)^{2}\right]+
\]
\[
+2\sum_{\substack{i,j=1\\
i<j
}
}^{n}\mathsf{cov}\left[\left(x_{i}^{(a_{1}a_{2}\cdots a_{N})}-x_{i}^{(b_{1}b_{2}\cdots b_{N})}\right)^{2},\left(x_{j}^{(a_{1}a_{2}\cdots a_{N})}-x_{j}^{(b_{1}b_{2}\cdots b_{N})}\right)^{2}\right]
\]
Let us denote $x_{i}^{(a_{1}a_{2}\cdots a_{N})}\equiv x_{i}$, $x_{i}^{(b_{1}b_{2}\cdots b_{N})}\equiv y_{i}$.
Then
\[
\sum_{i=1}^{n}\mathsf{Var}\left[\left(x_{i}^{(a_{1}a_{2}\cdots a_{N})}-x_{i}^{(b_{1}b_{2}\cdots b_{N})}\right)^{2}\right]=\sum_{i=1}^{n}\mathsf{Var}\left[\left(x_{i}-y_{i}\right)^{2}\right]=
\]
\[
=\sum_{i=1}^{n}\left(\mathsf{M}\left[x_{i}^{4}\right]+\mathsf{M}\left[y_{i}^{4}\right]-4\mathsf{M}\left[x_{i}^{3}\right]\mathsf{M}\left[y_{i}\right]-4\mathsf{M}\left[y_{i}^{3}\right]\mathsf{M}\left[x_{i}\right]\right)+
\]
\[
+\sum_{i=1}^{n}\left(4\mathsf{M}\left[x_{i}^{2}\right]\mathsf{M}\left[y_{i}^{2}\right]-\left(\mathsf{M}\left[x_{i}^{2}\right]\right)^{2}-\left(\mathsf{M}\left[y_{i}^{2}\right]\right)^{2}\right)+
\]
\[
+\sum_{i=1}^{n}\left(4\mathsf{M}\left[x_{i}^{2}\right]\mathsf{M}
\left[x_{i}\right]\mathsf{M}\left[y_{i}\right]+4\mathsf{M}\left[y_{i}^{2}\right]
\mathsf{M}\left[y_{i}\right]\mathsf{M}\left[x_{i}\right]\right),
\]
\[
\sum_{\substack{i,j=1\\
i<j
}
}^{n}\mathsf{cov}\left[\left(x_{i}^{(a_{1}a_{2}\cdots a_{N})}-x_{i}^{(b_{1}b_{2}\cdots b_{N})}\right)^{2},\left(x_{j}^{(a_{1}a_{2}\cdots a_{N})}-x_{j}^{(b_{1}b_{2}\cdots b_{N})}\right)^{2}\right]=
\]
\[
=\sum_{\substack
i,j=1\\
i<j}^{n}\mathsf{cov}\left[\left(x_{i}-y_{i}\right)^{2},\left(x_{j}-y_{j}\right)^{2}\right]=
\]
\[
=\sum_{\substack{i,j=1\\
i<j
}
}^{n}\left(\mathsf{cov}\left[x_{i}^{2},x_{j}^{2}\right]+\mathsf{cov}\left[y_{i}^{2},y_{j}^{2}\right]-
2\mathsf{M}\left[x_{i}\right]\mathsf{cov}\left[y_{j}^{2},y_{i}\right]-
2\mathsf{M}\left[x_{j}\right]\mathsf{cov}\left[y_{i}^{2},y_{j}\right]\right)+
\]
\[
-2\sum_{\substack
i,j=1\\
i<j}^{n}\left(\mathsf{M}\left[y_{i}\right]\mathsf{cov}
\left[x_{j}^{2},x_{i}\right]+\mathsf{M}\left[y_{j}\right]\mathsf{cov}\left[x_{i}^{2},x_{j}\right]\right)+
\]
\[
+4\sum_{\substack{i,j=1\\
i<j
}
}^{n}\left(\mathsf{cov}\left[x_{i},x_{j}\right]\mathsf{cov}\left[y_{i},y_{j}\right]+
\mathsf{cov}\left[x_{i},x_{j}\right]\mathsf{M}\left[y_{i}\right]\mathsf{M}\left[y_{j}\right]+
\mathsf{cov}\left[y_{i},y_{j}\right]\mathsf{M}\left[x_{i}\right]\mathsf{M}\left[x_{j}\right]\right).
\]
Since $\mathsf{M}\left[x_{i}\right]$, $\mathsf{M}\left[x_{i}^{2}\right]$,
$\mathsf{M}\left[x_{i}^{3}\right]$, $\mathsf{M}\left[x_{i}^{4}\right]$,
\emph{$\mathsf{cov}\left[x_{i},x_{j}\right]$} are uniformly bounded
and since
\[
\dfrac{1}{n^{2}}\sum_{\substack{i,j=1\\
i<j
}
}^{n}\mathsf{cov}\left[x_{i}^{2},x_{j}^{2}\right]\rightarrow0,\;\dfrac{1}{n^{2}}\sum_{\substack{i,j=1\\
i<j
}
}^{n}\mathsf{cov}\left[x_{i}^{2},x_{j}\right]\rightarrow0,\;\dfrac{1}{n^{2}}\sum_{\substack{i,j=1\\
i<j
}
}^{n}\mathsf{cov}\left[x_{i},x_{j}\right]\rightarrow0
\]
as $n\rightarrow\infty$ we get
\[
\dfrac{1}{n^{2}}\mathsf{Var}\left[\sum_{i=1}^{n}\left(x_{i}^{(a_{1}a_{2}\cdots a_{N})}-x_{i}^{(b_{1}b_{2}\cdots b_{N})}\right)^{2}\right]\rightarrow0.
\]
Then, by Markov's theorem,
\[
\dfrac{\sum_{i=1}^{n}\left(x_{i}^{(a_{1}a_{2}\cdots a_{N})}-x_{i}^{(b_{1}b_{2}\cdots b_{N})}\right)^{2}}{n}\overset{\mathrm{P}}{\longrightarrow}2\left(\left
(1-\delta_{a_{1}b_{1}}\delta_{a_{2}b_{2}}\cdots\delta_{a_{N}b_{N}}\right)\sigma_{N}^{2}+\right.
\]
\[
\left.+\left(1-\delta_{a_{1}b_{1}}\delta_{a_{2}b_{2}}\cdots\delta_{a_{N-1}b_{N-1}}\right)
\sigma_{N-1}^{2}+\cdots+\left(1-\delta_{a_{1}b_{1}}\right)\sigma_{1}^{2}\right).
\]
Next, by Slutsky's theorem,
\[
d_{n}\left(x_{i}^{(a_{1}a_{2}\cdots a_{N})}-x_{i}^{(b_{1}b_{2}\cdots b_{N})}\right)\overset{\mathrm{P}}{\longrightarrow}\sqrt{2}\left(\left(1-\delta_{a_{1}b_{1}}\delta_{a_{2}b_{2}}\cdots\delta_{a_{N}b_{N}}\right)\sigma_{N}^{2}+\right.
\]
\[
\left.+\left(1-\delta_{a_{1}b_{1}}\delta_{a_{2}b_{2}}\cdots\delta_{a_{N-1}b_{N-1}}\right)
\sigma_{N-1}^{2}+\cdots+\left(1-\delta_{a_{1}b_{1}}\right)\sigma_{1}^{2}\right)^{\tfrac{1}{2}}.
\]
This completes the proof of the main theorem.
\end{proof}

\section{Conclusions}

The present paper is an extension of the author's paper \cite{Z1},
in which a~procedure for constructing finite ultrametric spaces was proposed based on the generation of
a~finite number of independent random points with independent coordinates in~$\mathbb R^{n}$. It was shown that,
for a~special class of laws of distributions of points in~$\mathbb R^{n}$, the normalized matrix of Euclidean distances on the set of points converges
in probability as $n \to \infty$ to the ultrametric matrix.  In the present paper, we extend the result of~\cite{Z1} to the case when the coordinates of
random points are statistically dependent. Our main result is Theorem~\ref{th3} of Section~4, which states that, under a~number of conditions on the
expectations of the variance and the covariance of coordinates of random points, the matrix of Euclidean distances of a~random realization of a~finite number of points
tends in probability as $n \to \infty$ to the ultrametric matrix. The proof of Theorem~\ref{th3} depends, in particular, on the law of large numbers
in the form of Markov's theorem. We obtain the explicit form of this ultrametric matrix and show that it is completely determined by the expectations of
the conditional variances of the coordinates of points. The paper also contains two illustrative examples obtained by computer simulation of random points
in Euclidean spaces of large dimensions. These examples illustrate the working principle of the theorem for specific distribution of random points in two cases:
when the coordinates of points are dependent and when they are independent.

An interesting questions is how the mechanism of generation of the ultrametric considered above may manifest itself in
processing of data sets pertaining to real systems. In the present paper (as in the previous one) we do not pose the problem of
describing the real systems in which the above scenario of origination of ultrametric structures admits an exact realization. Nevertheless,
we may adduce some simple but general arguments supporting the idea that
the realization of the proposed scenario may take place in feature sets of real objects. Assume that we are given some set of
homogeneous objects. Assume that each object is assigned a~feature set which can be described by a~point in a~multivariate coordinate space  $\mathbb R^{n}$,
where $n$~is sufficiently large. We also assume that the feature sets are, in general, statistically independent random variables and that a~specific
distribution law of a~feature vector of each object depends on several external factors, which, in turn, depend upon a~certain set of random parameters.
We pose a~problem of classifying some sample of such object, whose solution is based upon comparing normalized Euclidean distances between the objects.
We suppose that a~specific realization of a~sample of objects is such that all objects from a~sample can be subdivided into subfamilies, which satisfy the following condition:
for each subfamily of objects all random parameters corresponding to external random factors have the same realization, whereas for distinct subfamilies of objects
the random parameters corresponding to external random factors have different realization. It is easily seen that under the above assumptions the principal conditions of
Theorem~\ref{th3} should be satisfied. Moreover, if the hypotheses on the expectations, variances and covariances of coordinates (features of objects) of Theorem~\ref{th3}
are also satisfied, then one may expect with large probability that the metric matrix of normalized Euclidean distances on the space of sample object features
is close to the ultrametric one. In this case, this will result in a~clusterization of objects from different subfamilies of the sample in terms of their proximity with respect to the
Euclidean metric.

It is especially noteworthy that in constructing the metric matrix on the set of random points in  $\mathbb R^{n}$ we actually used the Euclidean metric in~$\mathbb R^{n}$.
Nevertheless, such a~choice of a~metric is not unique.
For example, the distance between random points in  $\mathbb R^{n}$ can be measured with respect to the normalized Minkowski metric
$d(x,y)=\left (\frac{1}{n}\sum_{i=1}^{n}\left |x_i-y_i \right |^{\alpha}\right )^{\frac{1}{\alpha}}$,
which is the Hamming metric with $\alpha =1$, the Euclidean metric with $\alpha =2$, and the Chebyshev metric with $\alpha \to \infty$.
It is becomes an interesting question to determine the conditions for various~$\alpha$  imposed on the distribution of points in~$\mathbb R^{n}$
under which the metric matrix of their random realization will tend in probability to the ultrametric matrix as $n \to \infty$.
We hope to examine this question in the nearest future.

\vspace{5mm}
{\bf Acknowledgements}

The author is deeply indebted to Prof.\ I.\,V.~Volovich (Steklov Mathematical Institute,
Russian Academy of Sciences) and Prof.\ Fionn Murtagh (School of Computer Science and Informatics, De Montfort University)
for their  careful reading of the manuscript, discussion of the results, and a~number of useful comments.
The author is also grateful to Prof.\ A.\,R.~Alimov (Faculty of Mechanics and Mathematics,
Moscow State University) for his assistance in preparing the manuscript  and a~number of helpful comments.

\smallskip

The work was partially supported by the Ministry of Education and Science of of the Russian Federation under the Competitiveness Enhancement Program of SSAU for 2013--2020.





\bibliographystyle{elsarticle-num}



\end{document}